\documentclass{article}

\usepackage{arxiv}

\usepackage[utf8]{inputenc} 
\usepackage[T1]{fontenc}    
\usepackage{hyperref}       
\usepackage{url}            
\usepackage{booktabs}       
\usepackage{amsfonts}       
\usepackage{nicefrac}       
\usepackage{microtype}      
\usepackage{lipsum}

\usepackage{cite}
\usepackage{amsmath,amssymb,amsfonts}
\usepackage{graphicx}
\usepackage{algorithm}
\usepackage[noend]{algpseudocode}
\usepackage{hyperref}
\usepackage{textcomp}
\usepackage{threeparttablex} 

\usepackage{mathtools}
\usepackage{caption}
\usepackage{float}
\usepackage{stmaryrd}
\usepackage{epstopdf}
\usepackage{multirow}
\usepackage{pifont}
\usepackage{caption}
\usepackage{subcaption}
\usepackage{xcolor}

\newcommand{\norm}[1]{\left\lVert#1\right\rVert}

\newcommand{\R}{{\mathbb{R}}}

\newcommand{\N}{{\mathbb{N}}}

\newcommand{\dt}{\delta t}

\newcommand{\X}{{\mathbf{X}}}

\newcommand{\M}{{\mathcal{M}}}
\newcommand{\T}{{\mathbf{T}}}

\newcommand{\Obs}{{\mathcal{U}}}

\newcommand{\U}{{\mathbf{U}}}

\newtheorem{theorem}{Theorem}[section]

\newtheorem{assumption}{Assumption}

\newtheorem{definition}[theorem]{Definition}
\newtheorem{lemma}[theorem]{Lemma}
\newtheorem{remark}[theorem]{Remark}
\newtheorem{problem}[theorem]{Problem}
\newenvironment{proof}{\par\noindent\textbf{Proof.} }{\hfill$\blacksquare$\par}

\title{Smooth Spatiotemporal Tube Synthesis for Prescribed-Time  Reach-Avoid-Stay Control
\thanks{ This work was supported in part by the ARTPARK. The work of Ratnangshu Das was supported by the Prime Minister’s Research Fellowship from the Ministry of Education, Government of India.}
}

\author{
 Siddhartha Upadhyay \\
  Robert Bosch Centre for Cyber-Physical Systems\\
  IISc, Bengaluru, India\\
  \texttt{siddharthau@iisc.ac.in} \\
   \And
    Ratnangshu Das \\
  Robert Bosch Centre for Cyber-Physical Systems\\
  IISc, Bengaluru, India\\
  \texttt{ratnangshud@iisc.ac.in} \\
  \And
 Pushpak Jagtap \\
  Robert Bosch Centre for Cyber-Physical Systems\\
  IISc, Bengaluru, India\\
  \texttt{pushpak@iisc.ac.in} \\
}

\begin{document}
\maketitle

\begin{abstract}

In this work, we address the issue of controller synthesis for a control-affine nonlinear system to meet prescribed time reach-avoid-stay specifications. Our goal is to improve upon previous methods based on spatiotemporal tubes (STTs) by eliminating the need for circumvent functions, which often lead to abrupt tube modifications and high control effort. We propose an adaptive framework that constructs smooth STTs around static unsafe sets, enabling continuous avoidance while guiding the system toward the target within the prescribed time. A closed-form, approximation-free control law is derived to ensure the system trajectory remains within the tube and satisfies the RAS task. The effectiveness of the proposed approach is demonstrated through a case study, showing a significant reduction in control effort compared to prior methods.

\end{abstract}
\section{Introduction}
The reach–avoid–stay property (RAS) serves as a fundamental component for defining more complex temporal logic \cite{jagtap2024controller}. It supports the development of robust control strategies in safety-critical systems by ensuring that the system's states remain within a safe region, reach a designated target set within a finite time, and remain inside that target set thereafter.

To synthesis controllers for RAS specifications, two broad approaches exist. The first involves abstraction of the state space, while the other works directly in the continuous domain without discretization. With the growing use of formal languages to express complex tasks, symbolic control methods have gained prominence, as highlighted in \cite{tabuada2009verification}. For example, there are several tools \cite{rungger2016scots,jagtap2017quest} that synthesize controllers for nonlinear systems based on formal specifications like RAS. It does this by converting the continuous state space into a finite symbolic model through discretization. Similarly, in \cite{li2020robustly}, the author proposes a sound and complete control synthesis method for discrete-time nonlinear systems to verify the existence of a reach-and-stay control strategy and to construct the strategy if it exists, using s fixed-point algorithm. In \cite{sundarsingh2023scalable}, the authors present a scalable controller synthesis approach using Control Barrier Function (CBFs). However, despite the efforts to improve computational complexity, these approaches continue to face the curse of dimensionality.

In contrast to the abstraction-based approach, some approaches avoid discretization entirely and use barrier-based control \cite{ames2016control}, relying on optimization to compute the control signal. In \cite{meng2021control}, the authors formulate control Lyapunov-barrier functions that provide sufficient conditions for designing state-dependent control signals with respect to the RAS specifications. A probabilistic extension of this idea is presented in \cite{meng2022sufficient}, where the author connects probabilistic safety and stability guarantees with RAS specifications. However, their reliance on optimization makes them computationally expensive and difficult to scale to high-dimensional systems.

Another approach is funnel-based control, which is efficient and well-suited for tracking tasks \cite{bechlioulis2008robust}. This method is more computationally efficient, and a comprehensive overview of its applications is provided in \cite{bu2023prescribed}. In spite of its advantage in addressing tracking and reachability tasks, solving convex specifications, such as avoiding unsafe sets, remains a challenge.
In \cite{das2023funnel}, the funnel-based approach is extended to address RAS specifications by integrating avoid constraints through a circumvent function and an adaptive funnel framework. However, this approach still faces challenges, such as requiring accurate system models, being sensitive to disturbances, having higher computational costs, and not guaranteeing prescribed-time performance. 

To address these issues, \cite{STT} introduces a spatiotemporal tube (STT) framework that replaces traditional funnel functions. This method captures reach and avoid specifications simultaneously using smooth, time-varying tube and circumvent functions, while also ensuring prescribed-time satisfaction. A closed-form, approximation-free control law is derived to keep the system trajectory within the tube, thus fulfilling the prescribed-time RAS objectives. \cite{faruqui2025reach} presents a multi-agent coordination framework that uses STTs for prescribed-time RAS tasks under unknown dynamics and bounded disturbances. However, these approaches gave closed-form solutions to RAS taks, the use of circumvent functions often leads to sharp changes, resulting in high control effort. The smoother tube design using a sampling-based technique in \cite{das2024spatiotemporal, das2025approximation} reduces this effort but requires expensive offline computation.

In this work, we propose an adaptive framework for designing STTs in real time that enables smooth and efficient handling of avoid constraints. Unlike the circumvent-based approach, which introduces abrupt changes in the tube geometry and results in high control effort, the proposed method generates smooth, continuous deformations around unsafe sets. This significantly reduces the control effort, while preserving the guarantees of prescribed-time reach-avoid-stay satisfaction. The effectiveness of the proposed method is demonstrated through simulation case studies. 

The remainder of the paper is organized as follows: Section~\ref{sec:prelim} presents the preliminaries and problem formulation. In Section~\ref{sec:tube}, we describe the proposed adaptive STT design. Section~\ref{sec:control} presents the control law. In Section~\ref{sec:case}, we provide the simulation results on an omnidirectional mobile robot and compare the control effort with the circumvent-based approach. Finally, Section~\ref{sec:conc} concludes the paper and discusses future directions.

\section{Preliminaries and Problem Formulation}
\label{sec:prelim}
\subsection{Notation}
We use symbols $\N$, $\N_0$, $ \R$, $\R^+$, and $\R_0^+ $ to denote the set of natural, whole, real, positive real, and nonnegative real numbers, respectively. 
The space of the real $n \times m$ matrices is denoted by $\R^{n\times m}$. An $n$-dimensional column vector is represented by $\R^{n}$.
A vector $x \in \mathbb{R}^{n}$ with entries $x_1, \ldots, x_n\in \R$ is written as $[x_1, \ldots, x_n]^\top$, where $x_i \in \mathbb{R}$ is the $i$-th entry, and $i \in [1;n]$. The Euclidean norm is denoted by $\norm{\cdot}$, and the transpose of a matrix $M \in \R^{n \times m}$ by $M^\top$. 
Given sets $\X_i$, $i\in\left[1;N\right]$ with $N \in \N$, we denote the Cartesian product of the sets by $\X=\prod_{i\in\left[1;N\right]}\X_i:=\{(x_1,\ldots,x_N)|x_i\in \X_i,i\in\left[1;N\right]\}$. 
Consider a set $\X_a\subset\R^n$, its projection on $i$th dimension, where $i\in[1;n]$, is given by an interval $[\underline \X_{ai},\overline \X_{ai}]\subset \R$, where $\underline \X_{ai}:=\min\{x_i\in\R\mid[x_1\ldots,x_n]\in \X_a\}$ and $\overline \X_{ai}:=\max\{x_i\in\R\mid[x_1,\ldots,x_n]\in \X_a\}$. We further define the hyper-rectangle $\llbracket \X_a \rrbracket = \prod_{i=[1;n]}{[\underline \X_{ai}, \overline \X_{ai}]}$.
We denote the empty set by $\emptyset$.
The space of bounded continuous functions is denoted by $\mathcal{C}$, and $\mathcal{C}^n$ denotes the set of functions that are $n$ times continuously differentiable.

\subsection{System Definition}
Consider the following control-affine nonlinear system:
\begin{align}
    \mathcal{S}: \dot{x} = f(x) + g(x)u + w, \label{eqn:sysdyn}
\end{align}
where $x(t) = [x_1(t), \ldots, x_n(t)]^{\top} \in \mathbb{R}^n$ and $u(t) \in \mathbb{R}^n$ are the state and control input vectors, respectively. $w(t) \in \mathbb{W} \subset \R^n$ denotes unknown bounded disturbance. 

\begin{assumption}\label{assum:lip}
    The functions $f: \R^n\rightarrow \R^n$ and $g: \R^n\rightarrow \R^{n \times n}$ in \eqref{eqn:sysdyn} are \textit{unknown} and locally Lipschitz.         
\end{assumption}
\begin{assumption}\label{assum:pd}{(\cite{xu2003robust,PPCfeedback})}
    The symmetric components of $g(x)$ denoted by $g_s(x) = \frac{g(x)+g(x)^{\top}}{2}$ are uniformly sign definite with known signs for all $x \in \X$. Without loss of generality, we assume $g_s(x)$ is uniformly positive definite, i.e., there exists a constant $\underline{g}\in\mathbb R^+$:
    $0 < \underline{g} \leq \lambda_{\min} (g_s(x)), \forall x \in \X,$
    where $\lambda_{\min}(\cdot)$ is the smallest eigenvalue of the matrix.
\end{assumption}
\subsection{Problem Formulation}
We model the unsafe set as a union of compact and convex sets, $\mathbf{U}= \bigcup_{j \in [1;n_u]} \Obs^{(j)} \subset \R^n$. Note that, in general, $\mathbf{U}$ can be nonconvex and disconnected. The compact connected sets $\mathbf{S} \subset \R^n \setminus \mathbf{U}$ and $\mathbf{T}\subset \R^n \setminus \mathbf{U}$ represent the initial and target sets, respectively. 
We consider a prescribed-time reach-avoid-stay problem defined next.

\begin{definition}[Prescribed-time RAS Task]\label{def:ras}
Given an unsafe set $\mathbf{U}$, along with an initial set $\mathbf{S} \subset \R^n \setminus \mathbf{U}$ and a target set $\mathbf{T} \subset \R^n \setminus \mathbf{U}$, the goal is to ensure that, for any initial state $x(0) \in \mathbf{S}$, there exists a time $t \in [0, t_c]$ such that $x(t) \in \mathbf{T}$, and the trajectory avoids the unsafe set for all times, i.e., $x(s) \in \R^n \setminus \mathbf{U}$ for all $s \in [0, t_c]$, where $t_c \in \mathbb{R}^+$ is the prescribed completion time.
\end{definition}

\begin{problem}\label{prob1}
 Given the system $\mathcal{S}$ in \eqref{eqn:sysdyn} satisfying assumptions \ref{assum:lip} and \ref{assum:pd}, design a \textit{closed-form} and \textit{approximation-free} controller to ensure the satisfaction of prescribed-time reach-avoid-stay specifications in Definition \ref{def:ras}.
\end{problem}

To address the Problem \ref{prob1}, we employ the STT framework as introduce in \cite{STT}, which defines a time-dependent region in the state space that ensures safety and maintains progress toward the goal over the entire time horizon.
\begin{definition}\label{tube_def}
Given a Prescribed-time RAS  specification in Definition \ref{def:ras}, a spatiotemporal tube (STT), $\Gamma(t)=\prod_{i\in\left[1;N\right]}[\gamma_{i,L}(t),\gamma_{i,U}(t)],\forall t \in \mathbb{R}_0^+$ is characterized by continuously differentiable time varying functions $\gamma_{i,L}:\mathbb{R}_0^+ \rightarrow \mathbb{R} \text{ and }\gamma_{i,U}:\mathbb{R}_0^+ \rightarrow \mathbb{R}, \forall i \in [1;n]$ if the following holds:
\begin{align}\label{eqn:stt}
    &\gamma_{i,L}(t)<\gamma_{i,U}(t),\forall i \in[1;n],t\in[0,t_c],\\
    &\Gamma(0) \subseteq \mathbf{S},\
    \Gamma(t_c) \subseteq \T,\
    \Gamma(t) \cap \mathbf{U}= \emptyset, \forall t \in [0,t_c]. \notag
\end{align}
\end{definition}

\section{Designing Spatiotemporal Tubes} \label{sec:tube}

Unlike the approach in \cite{STT}, which uses a circumvent function to sharply modify the tube around unsafe sets, our method constructs the STTs that smoothly adapts to avoid these sets while still satisfying the conditions in Definition~\ref{tube_def}. This leads to a smoother and more consistent tube design without abrupt changes.

\subsection{Reachability Tubes}

In this section, we introduce the concept of reachability tubes as in \cite{das2024spatiotemporal} that ensure the system trajectory, starting from any initial state $x(0) \in \R^n\backslash \U$, reaches the designated target set $\mathbf{T}$ within a specified finite time $t_c$. 

We begin by defining a hyperrectangle $\hat{\mathbf{S}}$ centered around the initial state $x(0)\in \operatorname{int}(\mathbf{S}) \subset \mathbf{S}$ as:
\begin{align}\label{eqn:initial reg}
    \hat{\mathbf{S}}:=\prod_{i=[1;n]}[x_i(0)-d_{i,S},x_i(0)+d_{i,S}]\subset \mathbf{S_0},
\end{align}
where $d_{S,i}\in \mathbb{R}^+$ represents the extent of the hyperrectangle $\hat{\mathbf{S}}$ in the $i$-th dimension. 

Similarly, a reference point $\eta=[\eta_1,...,\eta_n]\in \operatorname{int}(\mathbf{T})$ can be chosen, and based on this, a hyperrectangle $\hat{\mathbf{T}}\subset \mathbf{T}$ centered at $\eta$ is defined as:
\begin{align}
    \hat{\mathbf{T}}:=\prod_{i=[1;n]}[\eta_i-d_{i,T},\eta_i+d_{i,T}]\subset \mathbf{T},
\end{align}
where $d_{i,T}\in \mathbb{R}^+$ determines the size of $\hat{\mathbf{T}}$ in each dimension.

{
The reachability tube margin $\rho(t)=[\rho_{1}(t),....,\rho_{n}(t)]$ defines the lower bound of the tube along each dimension and evolves according to the following differential equation:
\begin{align}
\dot{\rho}_{i}(t) &=
\begin{cases}
\displaystyle t_c\frac{\hat{\underline{\mathbf{T}}}_i - \hat{\underline{\mathbf{S}}}_{i} }{(t_c - t)^2}  \operatorname{sech}^2\left( \frac{t}{t_c - t} \right),  &\text{if } t < t_c \\
0, &\text{if } t \geq t_c
\end{cases}.
\end{align}
}

\begin{figure*}
    \centering
    \includegraphics[width=0.9\linewidth]{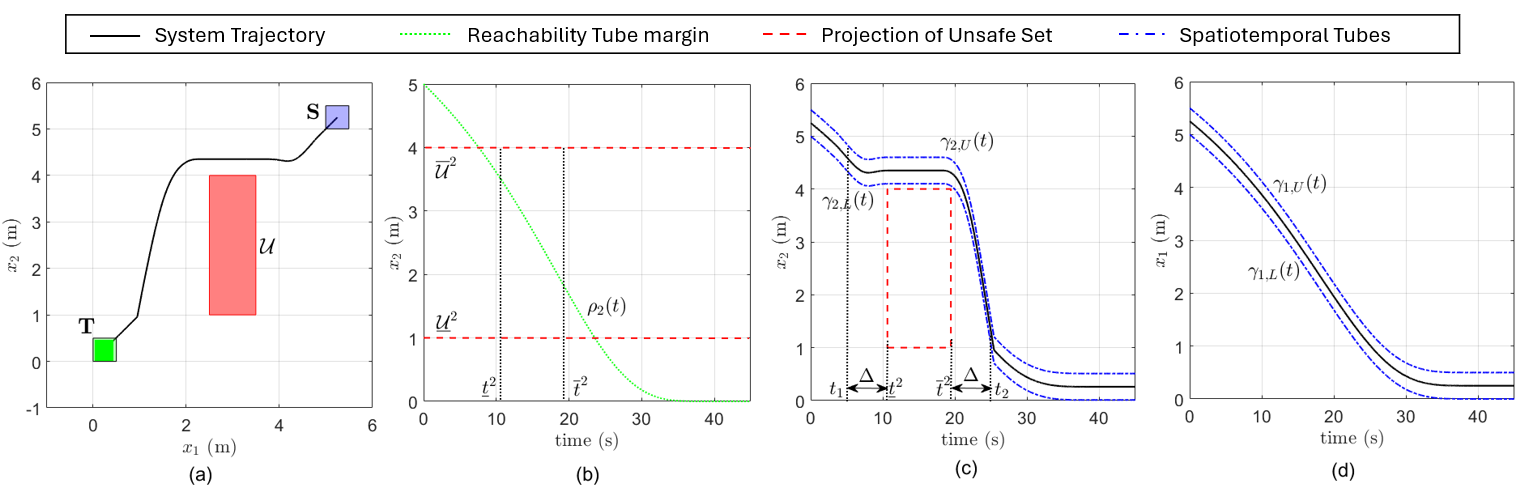}
    \caption{
An example of STT design for a two-dimensional system $x(t) = [x_1(t), x_2(t)]^\top$, assigned the task: “Starting from $\mathbf{S}$ (blue), reach $\mathbf{T}$ (green), while avoiding $\mathbf{U}$ (red), within $t_c = 45s$ ”. 
(a) State space with system trajectory (black dotted line). 
(b) Reachability tube margin along with the unsafe set projection in $x_2$. 
(c) STT adapted around the unsafe set projection in $x_2$. 
(d) STT in $x_1$.
}
    \label{fig:main}
\end{figure*}
\subsection{Avoid Specification}
Let the unsafe set be defined as \( \mathbf{U} = \bigcup_{j \in [1, n_u]} \mathcal{U}^{(j)} \subset \R^n \), which satisfies Assumption \ref{obs_assum}. In this formulation, instead of explicitly employing a circumvent function in an adaptive framework as in \cite{STT}, we enforce the avoidance specification by utilizing the time of intersection with the unsafe set and its corresponding projection $\llbracket \Obs^{(j)} \rrbracket = \prod_{i \in [1;n]} [\underline{\Obs}^{(j)}_i, \overline{\Obs}^{(j)}_i]$.

\begin{assumption}[\cite{STT}]\label{obs_assum}
    There exits at least one dimension $i\in[1;n]$ in which $[\underline{\hat{\mathbf{S}}}_{i},\hat{\overline{\mathbf{S}}}_{i}]\cap [\underline{\mathcal{U}}^{(j)}_{i},\overline{{\mathcal{U}}}^{(j)}_{i}] = \emptyset, \forall j\in[1;n_u]$. Similarly, there exists at least one dimension $i\in[1;n]$ in which $[\underline{\hat{\mathbf{T}}}_{i},\hat {\overline{{\mathbf{T}}}}_{i}]\cap [\underline{\mathcal{U}}^{(j)}_{i},\overline{{\mathcal{U}}}^{(j)}_{i}] = \emptyset $, $\forall j\in[1;n_u]$.
\end{assumption}

Next, we will use the formulation of reachability tube  to find the time of intersection $[\underline{t}^{(j)},\overline{t}^{(j)}]$ over which the tube starting from $\mathbf{S}$ to $\mathbf{T}$ intersects with the $j$th unsafe set $\mathcal{U}^{(j)}$ and is given by
\begin{align}
\underline{t}^{j} &= \left( \max_{i \in [1;n]} \min \bar{a}_i \right) t_c, \quad
\overline{t}^{j} = \left( \min_{i \in [1;n]} \max \bar{a}_i \right) t_c
\label{eq:time_bounds}
\end{align}
where $\bar{a}_i = \left\{ 
\frac{a^{j}_{1,i}}{1 + a^{j}_{1,i}},  
\frac{a^{j}_{2,i}}{1 + a^{j}_{2,i}},  
\frac{a^{j}_{3,i}}{1 + a^{j}_{3,i}},  
\frac{a^{j}_{4,i}}{1 + a^{j}_{4,i}} 
\right\}$ with $a^{j}_{1,i} = \tanh^{-1} \left( \frac{\underline{\mathcal{U}}^{j}_{i} - \hat{\underline{\mathbf{S}}}_{i}}{\hat{\underline{\mathbf{T}}}_{i} - \hat{\underline{\mathbf{S}}}_{i}} \right), a^{j}_{2,i} = \tanh^{-1} \left( \frac{\overline{\mathcal{U}}^{j}_{i} - \hat{\underline{\mathbf{S}}}_{i}}{\hat{\underline{\mathbf{T}}}_{i} - \hat{\underline{\mathbf{S}}}_{i}} \right),\\ a^{j}_{3,i} = \tanh^{-1} \left( \frac{\underline{\mathcal{U}}^{j}_{i} - \overline{\hat{\mathbf{S}}}_{i}}{\overline{\hat{\mathbf{T}}}_{i} - \overline{\hat{\mathbf{S}}_{i}}} \right)$, and $a^{j}_{4,i} = \tanh^{-1} \left( \frac{\overline{\mathcal{U}}^{j}_{i} - \overline{\hat{\mathbf{S}}}_{i}}{\overline{\hat{\mathbf{T}}}_{i} - \overline{\hat{\mathbf{S}}}_{i}} \right)$.

 In addition, it is worth noting that if the system trajectory enters the unsafe set $\mathcal{U}^{(j)}$, then $\exists t \in \mathbb{R}^+$ such that $x_i(t) \cap [\underline{\mathcal{U}}^{(j)},\overline{\mathcal{U}}^{(j)}]\neq \emptyset, \forall i \in [1;n]$. Thus, in order to satisfy the avoid specification, it is enough to modify the tube design in any one of the dimensions $i^{(j)}\in [1;n]$ given by:
 \begin{align}\label{eq:argmin_argmax}
i_1^{(j)} &= \operatorname{argmax}_{i \in [1;n]} \min \overline{a}_i, \quad
i_2^{(j)} = \operatorname{argmin}_{i \in [1;n]} \max \overline{a}_i, \nonumber \\
i^{(j)} &= \operatorname{argmin}_{k \in \{i_1^{(j)}, i_2^{(j)}\}} \left( \max \overline{a}_k - \min \overline{a}_k \right).
\end{align}

\subsection{Modifying the Tube}
In this section, we outline the approach that will later be extended into Algorithm 1 for modifying the STT to guarantee adherence to the avoidance specification.
For each unsafe set $\mathcal{U}^{(j)}$, we assume that the corresponding time interval of intersection, denoted by $[\underline{t}^{(j)}, \overline{t}^{(j)}]$, has already been identified as previously described. We define $t_1^{(j)} = \underline{t}^{(j)} - \Delta$ and $t_2^{(j)} = \overline{t}^{(j)} + \Delta$, where $\Delta \in \mathbb{R}^+$ represents a margin around the intersection interval.
{
\begin{assumption}\label{ass:unsafe_temp}
    All unsafe sets \(\mathcal{U}^{(j)} \in \U\) are temporally well-separated, i.e., for all $j\in[1;n_u]$ and for all $m\in[1;n_u]\setminus \{j\}$
     \begin{align}
        \min(|\underline{t}^{(j)}-\overline{t}^{(m)}|,|\overline{t}^{(j)}-\underline{t}^{(m)}|) > 2\Delta.
    \end{align}
\end{assumption}
}
{
\begin{remark}
Assumption~\ref{ass:unsafe_temp} is typically satisfied by task design. Specifically, if the prescribed time to reach the target is long enough, the unsafe sets become more spread out in time, making it easier to maintain the required buffer $2\Delta$ and ensure safe movement. 
\end{remark}
}
As discussed earlier, the tube can be modified along any dimension $i \in [1, n]$, and we refer to the chosen dimension for the $j^\text{th}$ unsafe set as $k = i^{(j)}$. Consequently, the task simplifies to adjusting the tube along the selected dimension so that it steers clear of the circumvent function, which was previously defined using a step function. To achieve this, we propose the following adaptive framework:

\begin{align}
&\dot\gamma_{i,L}(t)=
\begin{cases}
\dot\rho_{i}(t), & \text{if } i \neq k \\
\alpha_1^{(j)} \dot\rho_{i}(t) + \alpha_2^{(j)} \phi_{1}^{(j)}(t) + \alpha_3^{(j)} \phi_2^{(j)}(t), & \text{if } i = k.
\end{cases}\label{eqn:gamma_U}
\end{align}


In \eqref{eqn:gamma_U}, the weighting terms, which are a function of time, are defined using $s(t)=0.5\tanh \left( \frac{t}{v} \right),$ with  $v \in \mathbb{R}^+_0$, as:
\begin{align}
    \alpha_1^{(j)} & = s(t)-s(t-t_1^{(j)}+\dt)+s(t-t_2^{(j)}-\dt)+0.5,\nonumber\\ 
    \alpha_2^{(j)} & = s(t - t_1^{(j)}+\dt) - s(t - \underline{t}^{(j)}-\dt),\\
    \alpha_3^{(j)} & = s(t - \overline{t}^{(j)}+\dt)-s(t-t_2^{(j)}-\dt),\nonumber
\end{align}
where $\dt \in \R^+$, with $\dt < \Delta$, is a small buffer introduced to compensate for the smooth $\tanh$ based approximation in $s(t)$. 
The adaptive function $\phi_{1}^{(j)}(t)$ is defined as:
\[
\phi_{1}^{(j)}(t) := \frac{h^{(j)}_{1}(t) - \gamma_{i,L}(t)}{\underline{t}^{(j)} - t}.
\]
\[
h_1^{(j)}(t) :=
\begin{cases}
\begin{aligned}
&\big( \psi_{i,\M}^{(j)}  - \rho_{i}(t_1^{(j)}) \big) \tanh\!\left( \frac{t - t_1^{(j)}}{\underline{t}^{(j)} - t} \right) + \rho_{i}(t_1^{(j)}), 
\end{aligned} & t < \underline{t}^{(j)} \\
\psi_{i,\M}^{(j)} , & t \geq \underline{t}^{(j)}.
\end{cases}
\]
The adaptive function $\phi_{2}^{(j)}(t)$ is defined as:
\begin{equation}
\phi_{2}^{(j)}(t) := \frac{h^{(j)}_2(t) - \gamma_{i,L}(t)}{{t}_2^{(j)} - t}
\end{equation}
\[
h_2^{(j)}(t) :=
\begin{cases}
\begin{aligned}
&\big( \rho_{i}(t_2^{(j)}) - \psi_{i,\M}^{(j)}  \big) 
\tanh\!\left( \dfrac{t - \overline{t}^{(j)}}{t_2^{(j)} - t} \right) + \psi_{i,\M}^{(j)} , 
\end{aligned} & t < t_2^{(j)} \\
\rho_{i}(t_2^{(j)}), & t \geq t_2^{(j)}
\end{cases}
\]

In the above formulation, $\M \in \{L, U\}$ indicates whether the lower or upper projection of the unsafe set \(\mathcal{U}^{(j)}\) in dimension \(i\) is being used. The projections of the unsafe set are defined as $\psi_{i,L}^{(j)} = \overline{\mathcal{U}}_i^{j}+d_U^{(j)}$ and \(\psi_{i,U}^{(j)} = \underline{\mathcal{U}}_i^{j} - 2\min(d_{S,i}, d_{T,i})-d_U^{(j)}\) where $d_U^{(j)}\in \mathbb{R}^+,\forall j \in[1;n_u]$. Finally, once the lower tube $\gamma_{i,L}(t)$ is modified as per the above framework, the corresponding upper tube can be computed as:
\begin{align}
    \gamma_{i,U}(t) = \gamma_{i,L}(t) + 2\min(d_{S,i}, d_{T,i}),\label{eqn:gamma_add}
\end{align}
where $d_{S,i}$ and $d_{T,i}$ represent the initial and terminal deviations in the $i^{\text{th}}$ dimension. 
\begin{remark}
The adaptive framework described in Equation~\eqref{eqn:gamma_U} operates through three distinct phases. Initially, in the absence of a nearby unsafe set, we have $\alpha_1^{(j)}= 1$, $\alpha_2^{(j)}= 0$, and $\alpha_3^{(j)}= 0$. In this phase, the first term dominates, guiding the tube toward the target set. As the system approaches an unsafe region, that is, during the interval $t\in[t_1^{(j)},\underline{t}^{(j)}]$, the values transition to $\alpha_1^{(j)}= 0$, $\alpha_2^{(j)}= 1$, and $\alpha_3^{(j)}= 0$, thus activating the second term, which modifies the tube to bend safely around the projection of the unsafe set in the selected dimension. Finally, for $t\in[\overline{t}^{(j)},t_2^{(j)}]$, we have $\alpha_1^{(j)}= 0$, $\alpha_2^{(j)}= 0$, and $\alpha_3^{(j)}= 1$, activating the third term to ensure that the modified tube smoothly transitions back to the original reachability margin. This entire adaptive modification process is illustrated in Figure~\ref{fig:main}~(a)–(d).

\end{remark}
 Thus, the reach-and-avoid specification can be enforced by designing a controller such that the state trajectory remains confined within a prescribed tube, expressed by the constraint:
\begin{align}\label{eqn:specs}
\gamma_{i,L}(t) < x_i(t) < \gamma_{i,U}(t), \quad \forall (t, i) \in [0,t_c] \times [1, n].
\end{align}


\begin{theorem}
    The STT $\Gamma(t)=\prod_{i=1,..n}[\gamma_{i,L},\gamma_{i,U}]$ as defined in \eqref{eqn:gamma_U} satisfies the following conditions:
    \begin{enumerate}
        \item[(i)] The tube starts from the initial set  $\mathbf{S}$ i.e., $\Gamma(0)\subset\mathbf{S}$.
        \item[(ii)] The tube reaches the target set within the prescribed time $t_c$: $\Gamma(t_c)\subset\mathbf{T}$.
        \item[(iii)] The tube avoids the unsafe set for all time $\Gamma(t)\cap \mathbf{U}=\emptyset $.
        \item[(iv)] The upper and lower boundary of the tube $\Gamma(t)$ satisfy the following condition: $\gamma_{i.L}(t)<\gamma_{i,U}(t),\forall i\in[1;n],\forall t \in [0,t_c].$
    \end{enumerate}
    \end{theorem}
    \begin{proof}
     We will prove all four claims in this proof:

        (i) At initial time $t=0$, we have $\alpha_{1}^{(j)}= 1$ and $\alpha_{2}^{(j)}, \alpha_{3}^{(j)} = 0, \forall j\in[1;n_u]$. Therefore, from \eqref{eqn:gamma_U}, for all $i\in [1;n]$:
        \begin{align*}
            &\gamma_{i,L}(0)=\rho_{i}(0)=\underline{\hat{\mathbf{S}}}_i\\
            &\gamma_{i,U}(0)=\gamma_{i,L}(0)+2\min(d_{S,i},d_{T,i})\\
             \implies &\gamma_{i,U}(0) \leq \underline{\hat{\mathbf{S}}}_i + d_{S,i} = \overline{\hat{\mathbf{S}}}_i \\
             \implies &[\gamma_{i,L}(0), \gamma_{i,U}(0)] \subset [\underline{\hat{\mathbf{S}}}_i, \overline{\hat{\mathbf{S}}}_i] \\
             \implies &\bigcap_{i \in [1;n]}[\gamma_{i,L}(0), \gamma_{i,U}(0)] = \Gamma(0) \subset \mathbf{S} = \bigcap_{i\in [1;n]} [\underline{\hat{\mathbf{S}}}_i, \overline{\hat{\mathbf{S}}}_i].
        \end{align*}
        (ii) At final time $t=t_c$, we have $\alpha_{1}^{(j)}= 1$ and $\alpha_{2}^{(j)}, \alpha_{3}^{(j)} = 0, \forall j\in[1;n_u]$. Therefore, from \eqref{eqn:gamma_U}, for all $i\in [1;n]$:
        \begin{align*}
            &\gamma_{i,L}(t_c)=\rho_{i}(t_c)=\underline{\hat{\mathbf{T}}}_i \\
            &\gamma_{i,U}(t_c)=\gamma_{i,L}(t_c)+\min(d_{S,i},d_{T,i})\\
             \implies &\gamma_{i,U}(t_c) \leq \underline{\hat{\mathbf{S}}}_i + d_{T,i} = \overline{\hat{\mathbf{T}}}_i \\
             \implies &[\gamma_{i,L}(t_c), \gamma_{i,U}(t_c)] \subset [\underline{\hat{\mathbf{T}}}_i, \overline{\hat{\mathbf{T}}}_i] \\
             \implies &\bigcap_{i \in [1;n]}[\gamma_{i,L}(t_c), \gamma_{i,U}(t_c)] = \Gamma(t_c) \subset \mathbf{T}= \bigcap_{i\in [1;n]}[\underline{\hat{\mathbf{T}}}_i, \overline{\hat{\mathbf{T}}}_i].
        \end{align*}
         (iii) To show that \(\Gamma(t) \cap \mathbf{U} = \emptyset\) for all \(t \in [0, t_c]\), it suffices to verify that there exists at least one dimension \(i \in [1, n]\) such that, for all \(j \in [1, n_u]\) and \(t \in [\underline{t}^{(j)}, \overline{t}^{(j)}]\), the following holds: \\
        \textbf{Case 1:} If \(\mathcal{M} = L\), then \(\gamma_{i,L}(t) > \overline{\mathcal{U}}^{(j)}_i\). \\
        \textbf{Case 2:} If \(\mathcal{M} = U\), then \(\gamma_{i,U}(t) < \underline{\mathcal{U}}^{(j)}_i\). \\
        For \(t \in [t^{(j)}_1, \underline{t}^{(j)}]\), we have \(\alpha_1^{(j)},\alpha_3^{(j)} = 0\) and \(\alpha_2^{(j)} = 1\). Substituting into Equation~\eqref{eqn:gamma_U} and evaluating at \(t = \underline{t}^{(j)}\), we obtain: \\
        \textbf{Case 1:} \(\gamma_{i,L}(\underline{t}^{(j)}) = \psi_{i,L}^{(j)} = \overline{\mathcal{U}}_i^{(j)} + d_U^{(j)} > \overline{\mathcal{U}}_i^{(j)}\). \\
        \textbf{Case 2:} \(\gamma_{i,U}(\underline{t}^{(j)}) = \psi_{i,U}^{(j)} = \underline{\mathcal{U}}_i^{(j)} - d_U^{(j)} < \underline{\mathcal{U}}_i^{(j)}\). \\
        Furthermore, since \(\dot{\gamma}_{i,L}(t) = 0\) for \(t \in [\underline{t}^{(j)}, \overline{t}^{(j)}]\), the function \(\gamma_{i,L}(t)\) (or \(\gamma_{i,U}(t)\)) remains constant over this interval. Hence, the required condition holds for the entire duration \(t \in [\underline{t}^{(j)}, \overline{t}^{(j)}]\).
         \\
    (iv) The final part of the proof follows directly from \eqref{eqn:gamma_add}, as $ \min(d_{S,i}, d_{T,i}) \in \mathbb{R}^+ $, allowing us to conclude that:
    \begin{align*}
        \gamma_{i,U}(t) &= \gamma_{i,L}(t) + 2\min(d_{S,i}, d_{T,i}) \\
        &> \gamma_{i,L}(t)<\gamma_{i,U},\forall i \in[1;n] , \forall t\in[0,t_c].
    \end{align*}
\end{proof}

\begin{lemma}\label{lemma1}
    The functions $\gamma_{i,L}(t)$, $\gamma_{i,U}(t)$ and their derivatives are continuous and bounded for all $t \in [0, t_c]$.
\end{lemma}
\begin{proof}
  For all dimensions $i \neq k$, the tube boundaries $\gamma_{i,L}(t)$ and $\gamma_{i,U}(t)$, along with their derivatives $\dot\gamma_{i,L}(t)$ and $\dot\gamma_{i,U}(t)$, are continuous and bounded since they are constructed using the hyperbolic tangent function $\tanh$, which is smooth. For the critical dimension $i = k$, each term in \ref{eqn:gamma_U} is a product of functions that are assumed to be continuously differentiable and bounded. Specifically, the weighting functions $\alpha_m^{(j)},m\in\{1,2,3\}$ are constructed for smooth activation profiles and are therefore $\mathcal{C}^1$ and bounded. The functions $\dot\rho_{i}(t)$, $\phi_1^{(j)}(t)$, and $\phi_2^{(j)}(t)$ are also designed to be continuous and bounded. Since sums and products of continuous and bounded functions remain continuous and bounded, it follows that $\dot\gamma_{i,L}(t)$ is continuous and bounded and thus ${\gamma}_{i,L}(t)$ is continuous and bounded for $t\in[0,t_c]$ and also from \eqref{eqn:gamma_add} $\dot\gamma_{i,U}(t)$ and $\dot\gamma_{i,U}(t)$ are bounded and continuous for $t\in[0,t_c]$.
\end{proof}

\section{Controller Design}\label{sec:control}

In this section, we utilize the STTs described in \eqref{eqn:specs} to construct an explicit, approximation-free control law. Our method follows the structure presented in \cite{STT} and ensures compliance with the state constraints specified in \eqref{eqn:specs}.

We begin by introducing the normalized error $e(x,t)$, the transformed error $\varepsilon(x,t)$, and the diagonalized matrix $\xi(x,t)$ as defined below:
\begin{align}
    e(x,t) &= [e_1(x_1,t),\ldots,e_n(x_n,t)]^\top = \gamma_d^{-1}(t)(2x - \gamma_s(t)), \notag \\
    \varepsilon(x,t) &= \left[\ln\left(\frac{1 + e_1(x_1,t)}{1 - e_1(x_1,t)}\right), \ldots, \ln\left(\frac{1 + e_n(x_n,t)}{1 - e_n(x_n,t)}\right)\right]^\top, \nonumber \\
    \xi(x,t) &= 4\text{diag} \Big( \Big[ \frac{1}{\gamma_{1,d}(1-e_1^2(x_1,t))}, \ldots,  \frac{1}{\gamma_{n,d}(1-e_n^2(x_n,t))} \Big] \Big),
\end{align}
where $\gamma_s := [\gamma_{1,U} + \gamma_{1,L}, \ldots, \gamma_{n,U} + \gamma_{n,L}]$ and $\gamma_d := \operatorname{diag}(\gamma_{1,d}, \ldots, \gamma_{n,d})$, with $\gamma_{i,d} = \gamma_{i,U} - \gamma_{i,L}$.

\begin{theorem}
Consider the nonlinear control-affine system $S$ defined in \eqref{eqn:sysdyn}. Suppose the initial condition $x(0)$ lies within the prescribed STTs, i.e., $\gamma_{i,L}(0) < x_i(0) < \gamma_{i,U}(0)$ for all $i \in [1;n]$. Then, the following closed-form control law:
\begin{align}\label{eq:control}
    u(x,t) = -\kappa \xi(x,t)\varepsilon(x,t), \quad \kappa \in \mathbb{R}^+
\end{align}
\end{theorem}
 ensures the state trajectory satisfies the constraint in \eqref{eqn:specs}, steering $x(t)$ to the target set $\T$ within finite time $t_c$, while safely avoiding the unsafe region $\mathbf{U}$ and $\kappa$ is the user defined control gain.

\begin{proof}
    Following the proof of Theorem 1 in \cite{STT}, the control law in Equation \eqref{eq:control} ensures that the system trajectory remains within the STTs over $[0, t_c]$. 
\end{proof}

\begin{remark}
    It is important to highlight that the proposed control law is approximation-free and ensures the satisfaction of the RAS specification for control-affine systems with unknown dynamics. Moreover, if $g_s(x)$ is negative definite, the control gain $k$ must belong to the set $\mathbb{R} \setminus \mathbb{R}_0^+$.
\end{remark}

\begin{algorithm}
\caption{Adaptive Framework for STT}
\begin{algorithmic}[1]\label{Alg1}
\State \textbf{Input:} Initial state $x(0)$, target set $\T$, unsafe set $\mathbf{U} = \{\mathcal{U}_1, \ldots, \mathcal{U}_{n_u}\}$, $t_c$
\State \textbf{Output:} Control input $u(x(0), \T, \mathbf{U}, x(t), t)$ for all $t \in [0, t_c]$ in order to satisfy the given specification

\While{$t \neq t_c$}
   \State Determine $\mathbb{T}_{int}^{(j)}=[\underline{t}^{(j)},\overline{t}^{(j)}],$ {$ \forall     \text{ }\mathcal{U}_j \in \mathbf{U}$} as shown in \eqref{eq:time_bounds}
    \State Select the $\mathbb{T}_{int}^{(j)}$ with the smallest $\underline{t}^{(j)}$.
    \State Choose dimension $i^j \in \{1, \ldots, n\}$ using \eqref{eq:argmin_argmax}.
    \State Apply the adaptive framework to compute the STTs $\{\gamma_{i,L}(t),\gamma_{i,U}(t)\},\forall i\in[1;n]$ as given in \eqref{eqn:gamma_U}.
    \State Apply control law based on the STTs according to \eqref{eq:control}.
    
\EndWhile
\end{algorithmic}
\end{algorithm}

\section{Case Studies} \label{sec:case}

To demonstrate the effectiveness of the proposed approach, we present a case study in a 2-D arena simulated in \textsc{Matlab} with the omnidirectional robot dynamics adopted from \cite{jagtap2024controller}.
\begin{align}\label{eq:omni_dyn}
\begin{bmatrix}
\dot{x}_1 \\
\dot{x}_2 \\
\dot{x}_3
\end{bmatrix}
&=
\begin{bmatrix}
\cos x_3 & -\sin x_3 & 0 \\
\sin x_3 & \cos x_3  & 0 \\
0        & 0         & 1
\end{bmatrix}
\begin{bmatrix}
v_1 \\
v_2 \\
\omega
\end{bmatrix}
+ w(t).
\end{align}
Here, the state vector $[x_1, x_2, x_3]^T$ represents the robot’s pose, while $[v_1, v_2, \omega]^T$ denotes the velocity inputs. The term $w(t)$ denotes an unknown but bounded external disturbance. 
The robot navigates in a 2-D arena with multiple obstacles $\mathcal{U}^1 = [1.5, 2] \times [0.5, 3]$, $\mathcal{U}^2 = [5.2, 6.8] \times [3.2, 4]$ and $\mathcal{U}^3 = [7, 8] \times [0, 8]$, as illustrated in Figure~\ref{fig:main_t}. The robot's objective is to reach the target $\mathbf{T} = [11, 11.5] \times [7, 7.5]$ from the initial region $\mathbf{S} = [0, 0.5] \times [0, 0.5]$, within a prescribed time $t_c = 80$s, while safely avoiding the obstacles. The system trajectory along with the STTs is illustrated in Figure~\ref{fig:main_t}. Figure~\ref{fig:control} compares the control effort of the proposed smooth tube synthesis method with the circumvent-based approach \cite{STT}, highlighting a significant reduction in control effort achieved by our method.
    
    


\begin{figure*}
    \centering
    \includegraphics[width=\linewidth]{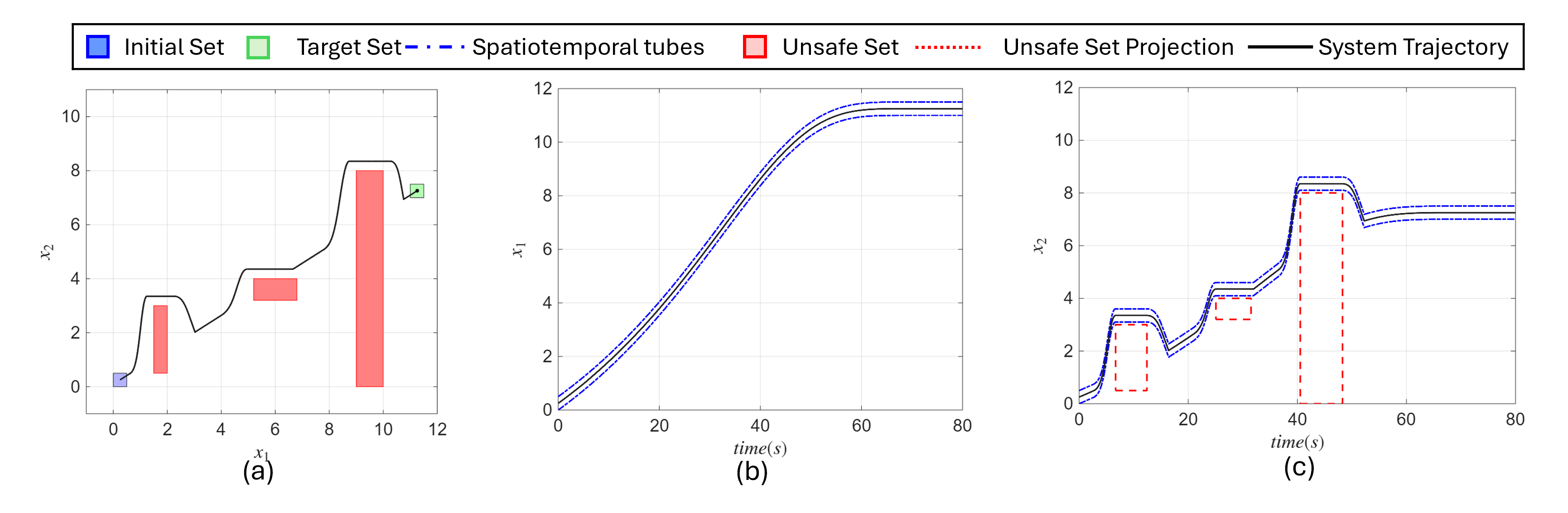}
    \caption{Simulation results showing the system trajectory (in black solid lines) followed by a robot with omnidirectional dynamics along with the STTs in dimension $x_1$ and $x_2$.}
    \label{fig:main_t}
\end{figure*}

\begin{figure}
    \centering
    \includegraphics[width=0.6\linewidth]{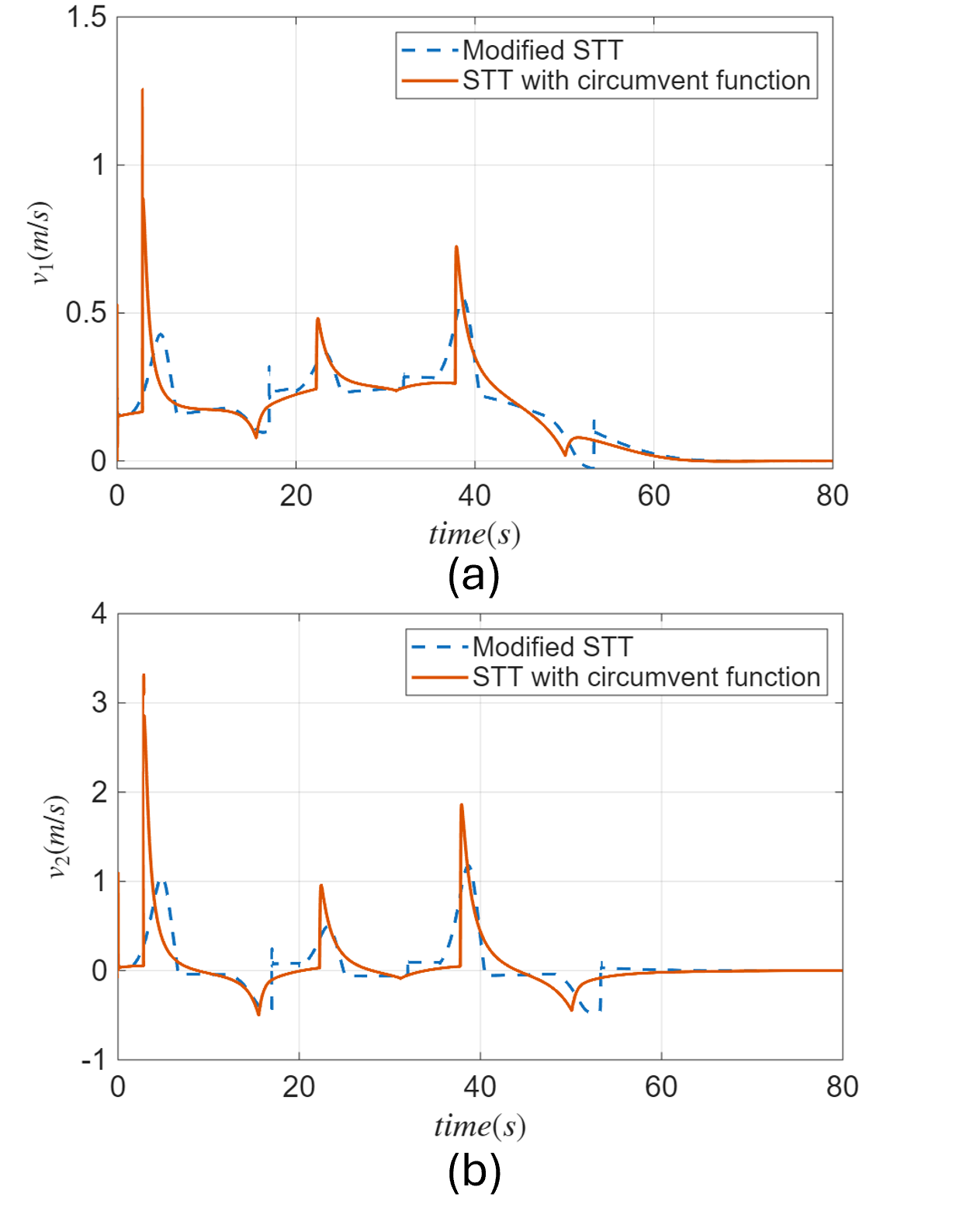}
    \caption{Comparing control effort of STT with circumvent function and modified STT framework  }
    \label{fig:control}
\end{figure}
    
\section{Conclusion}\label{sec:conc}

In this work, we presented a modified  STT-based framework to address the prescribed-time reach-avoid-stay control problem for systems. We have shown through a comparison study that our proposed framework enables smooth and efficient handling of avoid constraints and helps to lower the control effort when compared to circumvent based approach. We have also shown that the STTs designed by our proposed framework satisfies the given PT-RAS specification and along with that an closed-form, approximation-free feedback control law is also presented for control-affine systems with unknown dynamics, ensuring that system trajectories remain within the constructed tubes and satisfy the specified temporal and safety requirements. The effectiveness of the proposed method was validated through a robot navigation case study.


While the current work assumes that the unsafe sets are static and known, future research will focus on extending the framework to handle time-varying obstacles in dynamic environments. We also aim to generalize the method to broader classes of nonlinear systems.

\bibliographystyle{unsrt} 
\bibliography{sources} 

\end{document}